\documentclass[letterpaper, 10 pt, conference]{ieeeconf}  

\IEEEoverridecommandlockouts                              
\usepackage{graphics} 
\usepackage{epsfig} 
\usepackage{mathptmx} 
\usepackage{times} 

\usepackage{amsthm}
\usepackage{amsmath}
\usepackage{amssymb}  
\usepackage{amsfonts}
\usepackage{algpseudocode}
\usepackage{algorithm}
\usepackage{xcolor}
\usepackage{balance}
\usepackage{microtype}

\usepackage{cite}

\usepackage{graphicx}
\usepackage{float,epsf,caption,subcaption}

\newtheorem{theorem}{Theorem}

\newtheorem{corollary}{Corollary}

\theoremstyle{plain}

\theoremstyle{definition}
\newtheorem{definition}{Definition}
\newtheorem{example}{Example}
\theoremstyle{remark}
\newtheorem{remark}{Remark}

\newcommand{\R}{\mathbb{R}}
\newcommand{\END}{\hfill $\blacksquare$}

\makeatletter
\def\endthebibliography{%
	\def\@noitemerr{\@latex@warning{Empty 'thebibliography' environment}}%
	\endlist
}
\makeatother
\title{\LARGE \bf
On the Value of Preview Information For Safety Control
}

\author{Zexiang Liu, Necmiye Ozay
	\thanks{Zexiang Liu and Necmiye Ozay are with the Dept. of Electrical Engineering and Computer Science, Univ. of Michigan, Ann Arbor,	MI 48109, USA 
		{\tt\small zexiang,necmiye@umich.edu}. This work was supported by NSF Award
	CNS-1931982.}
}

\begin{document}
	
	\maketitle
	\thispagestyle{empty}
	\pagestyle{empty}

	\begin{abstract}
	Incorporating predictions of external inputs, which can otherwise be treated as disturbances, has been widely studied in control and computer science communities. These predictions are commonly referred to as preview in optimal control and lookahead in temporal logic synthesis. However, little work has been done for analyzing the value of preview information for safety control for systems with continuous state spaces. In this work, we start from showing general properties for discrete-time nonlinear systems with preview and strategies on how to determine a good preview time, and then we study a special class of linear systems, called systems in Brunovsky canonical form, and show special properties for this class of systems. In the end, we provide two numerical examples to further illustrate the value of preview in safety control.
	\end{abstract}

\section{Introduction}

 In a typical feedback control framework, the control input $u(t)$ is determined based on the current state $x(t)$, or more generally the initial state $x(0)$ and the sequence of the past disturbances\footnote{The concept of disturbance in this work can be quite general and it essentially captures any external input for which we might have predictions of future values. For instance, the reference signal in a tracking problem can be treated as ``disturbance" if error dynamics are used to include the reference signal in system equations (see examples in \cite{yu2020power, xu2019design}).} $d(0)$, $d(1)$, ..., $d(t-1)$.  However, in  this work, we allow $u(t)$ to be determined not only by $x(0)$, $d(0)$, ..., $d(t-1)$, but also by future disturbances $d(t)$, ..., $d(t+p)$, called the preview information, for some preview time $p$. This is a fair assumption in many modern control systems, enabled by the advances in sensing technologies. Examples of applying preview information in real-world systems include autonomous vehicles\cite{xu2019design}, power systems\cite{ozdemir2013design} and robotics\cite{kajita2003biped}.
 
 The above mentioned systems are all safety-critical, where controllers should be designed to ensure safety specifications. The safety specifications considered in this work are to have the system state avoid visiting a user-defined unsafe region, or equivalently have the state stay within a safe region indefinitely. A standard way to achieve safety in this sense is via robust controlled invariant sets\cite{bertsekas1972infinite, rungger2017computing}. Then, a fundamental question to ask is how to measure the improvement due to preview in safety control and how the change of preview time affects the quality of safety control.
 
 The majority of literature on preview control focuses on incorporating preview information into optimal control formulation\cite{sheridan1966three, tomizuka1975optimal, katayama1985design, xu2019design}. A prime example is model predictive control (MPC)\cite{garcia1989model, laks2011model, yu2020power}, where preview information is naturally incorporated into the state propagation constraints. In this case, the improvement due to preview is measured by the amount of cost reduction after increasing preview time. A recent work \cite{yu2020power} proves in theory that the cost reduction in both the linear quadratic control and MPC formulations decays exponentially fast as the preview time increases. However, those results are not applicable to our question, as they do not incorporate safety constraints.
 
 Our previous work addressed variants of this problem: \cite{liu2019safety} incorporates preview on mode switching into safety control of switched systems, and \cite{liu2020scalable} studies the structure of controlled invariant sets for linear systems with delay in input and preview in disturbance.  A significant implication of \cite{liu2020scalable} is that for linear systems, the negative impact of input delay to safety control can be compensated by the positive impact of preview on disturbances. But references \cite{liu2019safety, liu2020scalable} rather focus on algorithmic scalability and do not consider general systems. Therefore, they provide little theory in how different preview times affect the controlled invariant sets. 
 
 Notably, the impact of preview time is a relatively well-studied problem in reactive synthesis\cite{kupferman2011synthesis, holtmann2010degrees, klein2015much}, where preview is called lookahead. \cite{kupferman2011synthesis} provides, by checking the universal satisfiability of the linear temporal logic (LTL) formula encoding specifications, some extreme case analysis, which is analogous to our results on disturbance-collaborative systems in Section III. \cite{klein2015much} provides upper and lower bounds on the preview time necessary for the existence of a controller that realizes a LTL specification,  which sheds light on the impact of different preview times. But those results are for finite-state transition systems only. In our work, we are also interested in systems with continuous state spaces. 
 
 To summarize, to the best of our knowledge, there is little work in the literature that analyzes the value of preview for safety control of general discrete-time systems. This work is a first step in this direction. Our main contributions are: (i) We provide ways to compute inner and outer approximations of robust controlled invariant sets for general systems with preview and show how these approximation can be used to determine a good preview time. (ii) We derive a closed-form expression of the maximal controlled invariant set for systems in Brunovsky canonical form, one of the canonical forms of controllable systems, within a hyperbox safe set. Based on this closed-form expression, we characterize critical preview time over which additional preview information cannot improve safety.

In the remainder of this work, the preliminaries of controlled invariant sets and a formal definition of systems with preview are introduced in Section II. Then in Section III, we study analytical properties of the controlled invariant sets for general systems with preview and how those properties lead to strategies of selecting preview time. In Section IV, we develop the theory for systems in Brunovsky canonical form. 
After that, we illustrate the value of preview using two numerical examples in Section V and conclude the paper in Section VI. The proofs of the theorems and details of the examples can be found in Appendix.

\textbf{Notation:} For $K$ vectors ${x}_1\in \R^{n_{1}} $, ..., $ {x}_K\in \R^{n_{K}} $, we use $ ({x}_1, {x}_2, \cdots, x_{K}) $ or $x_{1:K} $ to denote their concatenation in $ \R^{ n_1 + n_2 \cdots + n_{K}} $. A single vector $ {x}\in \R^{n}$ can be also represented by $x = (x_1, x_2, \cdots, x_{n})$ where $x_{i}\in \R$ is the $i$ th entry of $x$.  We denote a closed interval between $a$ and $b$ by $[a,b]$. The sum of intervals $[a,b]$ and $[c,d]$ is denoted by $[a,b]+[c,d] = [a+c, b+d]$. Similarly, the subtraction of $[a,b]$ and $[c,d]$ is $[a,b]- [c,d] = [a-c,b-d]$. The sum of intervals $[a_{i}, b_{i}]$ for $i$ from $1$ to $n$ is denoted by $ \sum^{n}_{i=1} [a_{i}, b_{i}] = [ \sum^{n}_{i=1} a_{i}, \sum^{n}_{i=1} b_{i}]$. We also denote the sum and multiplication of a interval with a scalar by  $c+[a,b] = [c+a, c+b]$ and $ \alpha [a,b] = [ \alpha a, \alpha b]$ for $c\in \R$ and $ \alpha \geq 0$.  For either sum over scalars or intervals, we adopt the convention that $ \sum^{m}_{i=n} c_{i} =0 $ and {$ \sum^{m}_{i=n}[c_{i,1},c_{i,2}] = \emptyset $} if $n > m$. The Cartesian product of sets $X_1$, ..., $X_{n}$ are denoted by $X_1 \times \cdots \times X_{n}$ and/or $\Pi_{i=1}^{n} X_{i}$ and/or $X^{n}$ when $X_i = X$ for all $i$ from $1$ to $n$. A hyperbox $ \mathbf{B} = \{ {x} \in \R^{n} \mid x_{i}\in [c_{i, 1}, c_{i, 2}]\} $ in $\R^{n}$ is denoted by $ \mathbf{B} = \Pi_{i=1}^{n} [c_{i,1}, c_{i,2}]$.  Given a set $D \subseteq \R^{n}$ and a linear mapping $T: R^{n} \rightarrow R^{m}$, we denote the image of $D$ under $T$ by $TD = \{ T{x} \mid {x}\in D\} \subseteq R^{m} $.  Given a set $ X \subseteq  \R^{n}$, $PROJ_{j:k} (X) = \{ (x_{j}, \cdots, x_{k})  \mid (x_1, \cdots, x_{n})\in X \}$ is the projection of $X$ from $\R^{n}$ to the coordinates corresponding to $x_{j}$, ..., $x_{k}$ for $j$, $k$ with $1\leq j\leq k \leq n$.  

\section{Preliminaries} 
\label{sec:prelim} 

We consider discrete-time system $\Sigma$ in form of 
\begin{align} \label{eqn:sys} 
\Sigma: {x}(t+1) = f( {x}(t), u(t), d(t))
\end{align}
with state ${x}(t) \in \R^{n}$, control input $ {u}(t)\in \R^{m}$ and disturbance $ {d}(t) \in D\subseteq \R^{l}$.  
Let  $ S_{xu}\subseteq \R^{n+m} $ be the \emph{safe set} of $ \Sigma $ that describes safety constraints on the state-input pairs.
\begin{definition}\label{def:inv_set}
	A set $ C\subseteq \R^{n} $ is a \emph{controlled invariant set} of $\Sigma$ in safe set $ S_{xu}\subseteq \R^{n+m} $ if for all $ {x}\in C $, there exists some $  {u}\in \R^{m} $ such that $ ({x}, {u})\in S_{xu} $ and for all $ d\in D $,	$f({x}, {u}, {d})\in C$. $C_{max}$ is \emph{the maximal controlled invariant set} in $S_{xu}$ if $C_{max}$ contains any controlled invariant set of $\Sigma$ in  $S_{xu}$.
\end{definition}

For the remainder of this work, we use $C_{max}(\Sigma, S_{xu}) $ to denote the maximal controlled invariant set of system $\Sigma$ within safe set $S_{xu}$.  Given a controlled invariant set $C$, we define the \emph{admissible input set} at state ${x}$ by 
\begin{align}
\mathcal{A} (C,x) = \{ u\in \R^{m} \mid ({x},u)\in S_{xu}, f(x,u,d) \in C\}.
\end{align}
$\mathcal{A}(C,x)$ is the maximal admissible input set at $x$ when $C$ is the maximal controlled invariant set. If a set $ C $ is controlled invariant, there exists a safe controller $u_{safe}: C  \rightarrow \R^{m}$ such that any closed-loop trajectory starting from $ C $ stays in $ C $ indefinitely, robust to arbitrary disturbances in $ D $. 
A function $ u_{safe}: C \rightarrow \R^{m}	$ is a safety controller if and only if $u_{safe}(x) \in \mathcal{A}(C,x)$ for all $x\in C$. 

{In this work, we measure the conservativeness of controlled invariant sets by comparing (i) the size of the controlled invariant set, or (ii) the size of the admissible input set at a given state $ x $. There is a connection between these two measures:}
If we have controlled invariant sets $C_1$ and $C_2$ with $C_1 \subseteq C_2$, then $\mathcal{A}(C_1,x) \subseteq \mathcal{A}(C_2,x)$ for all $x\in C_1$.

To compute controlled invariant sets, we introduce the \emph{controlled predecessor operator} with respect to system $\Sigma$ as in \eqref{eqn:sys} 
\begin{align} \label{eqn:pre} 
Pre_{\Sigma}(X,S_{xu}) = \{ x \mid \exists u \text{ s.t. } (x,u)\in S_{xu},\nonumber \\
f(x,u,d) \in X, \forall d\in D\}. 
\end{align}

Define $X_0 = PROJ_{1:n} (S_{xu})$ and recursively define 
\begin{align}
X_{k} = Pre_{\Sigma} (X_{k-1},S_{xu}),\ k\geq 1. \label{eqn:pre_iterative}  
\end{align}
Under sufficient conditions in \cite{bertsekas1972infinite}, $ X_{k}$ converges to the maximal controlled invariant set $C_{max}(\Sigma,S_{xu})$.

\begin{definition}
	We call	a system $ \Sigma $ \emph{with $ p $-step preview} if the disturbances in the next $p$ steps can be measured at each time instant. In other words, the control input $u(t)$ at each time $t$ can be determined based on the state $ x(0) $ and the disturbances $ d(k) $ for $k$ from $ 0$ to $t+p-1$.
\end{definition}

For system $ \Sigma $ with $ p $-step preview, to explicitly indicate the available information on future disturbances at each time, we construct a \emph{p-augmented system} $ \Sigma_p $ with respect to system $\Sigma$ with state\footnote{We use $d_{1:p}(t)$ to denote the vector $(d_1(t), \cdots, d_{p}(t))$.} $ \xi(t)  := (x(t), d_{1:p}(t)) $, defined by
\begin{align} \label{eqn:sys_aug} 
\Sigma_{p}: \left\lbrace 
\begin{array}{rll}
x(t+1) &=& f(x(t),u(t),d_1(t))\\
d_{1}(t+1) &=& d_{2}(t)\\
&\cdots&\\
d_{p-1}(t+1) &=& d_{p}(t)\\
d_{p}(t+1) &=& d(t),
\end{array}\right.
\end{align}
with $\xi(t) \in \R^{n}\times D^{p} $, $ u(t)\in \R^m $ and $ d(t)\in D \subseteq \R^{l} $.  

Suppose $ \Sigma $ has safe set $ S_{xu} $. We define the \emph{$p$-augmented safe set} of $ \Sigma_{p} $ by $$ S_{xu,p} = \{ (x,d_{1:p}, u) \mid (x,u)\in S_{xu}, (d_{1}, \cdots, d_{p})\in D^{p}\}.$$ 
Note that {if $ d_{1:p} \in D^p$}, to check $(x, d_{1:p}, u) \in S_{xu,p}$ is equivalent to check $(x,u)\in S_{xu}$.
In what follows, we use $C_{max,p}(\Sigma, S_{xu})$ to denote the maximal controlled invariant set $C_{max}(\Sigma_{p}, S_{xu,p})$. When $\Sigma$ and $S_{xu}$ are clear from the context, we use $C_{max,p}$ for short. 

There are two baseline methods to compute controlled invariant sets of $\Sigma_{p}$ in $S_{xu,p}$:

\noindent \textbf{Method 1}: Apply the following iterative procedure: Compute $X_{0,p} = PROJ_{1:n}(S_{xu})\times D^{p}$. Then, compute $X_{k,p} = Pre_{\Sigma_{p}}(X_{k-1,p}, S_{xu,p})$ recursively until convergence, that is $X_{k+1,p} = X_{k,p}$. \END 

When the recursive procedure terminates, Method 1 returns the maximal controlled invariant set. However, Method 1 is not guaranteed to terminate in finite iterations and does not scale well for high-dimensional systems. Since the dimensionality $(n+pl)$ of $\Sigma_{p}$ is proportional to the preview time $p$, this method does not work well for systems with a long preview time.

\noindent \textbf{Method 2:} Find a conservative controlled invariant set $X_{0,p}$ of $\Sigma_{p}$ within $S_{xu,p}$. Pick a maximal iteration number $K \in \mathbb{N}\cup \{\infty\} $. Compute $X_{k,p} = Pre_{\Sigma_{p}}(X_{k-1,p}, S_{xu,p})$ recursively until $k\geq K$ or $X_{k,p} = X_{k-1,p}$. \END

A typical choice of $X_{0,p}$ is $C_{max}(\Sigma,S_{xu})\times D^{p}$. Since $ X_{0,p} $ is controlled invariant, the size of $X_{k,p}$ grows as $k$ increases, and $X_{k,p}$ is controlled invariant for any $k\geq 0$. In practice, Method 2 can be more scalable than Method 1. The drawback of Method 2 is that, as $K \rightarrow \infty	$,  $X_{K,p}$ does not necessarily converge to the maximal controlled invariant set, {as shown in Example \ref{exp:inside_out}}. In this sense, Method 2 is more conservative than Method 1. 

\begin{example} \label{exp:inside_out}   
	Consider the following $2$-dimensional system 
	\begin{align}
	\Sigma:  \begin{bmatrix}
	x_1(t+1)\\
	x_2(t+1)
	\end{bmatrix}= \begin{bmatrix}
	0 & 1\\
	0 & 0\end{bmatrix} + \begin{bmatrix}
	0\\1	
	\end{bmatrix}u(t) + 0\cdot d(t)
	\end{align}
	with $x(t)$, $ d(t)\in D \subseteq \R^{2}$ and $u(t)\in \R$. The safe set is  $S_{x}\times\R$, where $S_{x} = \{(a,a) \mid \vert a \vert \leq 1\}.$ Since the disturbance term is multiplied by $0$, for any preview horizon $p$, the maximal controlled invariant set of the $p$-augmented system is the $p$-augmented safe set $S_{x}\times D^{p} \times \R$. We apply Method 2 with the seed set $X_{0,p} = \{(0,0)\}\times D^{p}$. We can easily check that $X_{0,p}$ is controlled invariant. For arbitrarily large $K>0$ in Method 2, $X_{K,p} = X_{0,p}$ is strictly contained by the maximal controlled invariant set. \END
\end{example}

It is worth noting that if we use $C_{max}(\Sigma,S_{xu})$ as the terminal state constraints in a model predictive control formulation with the planning horizon $p$, {the feasible set of the {initial states and the disturbances is equal to the controlled invariant set obtained by taking $K=p$ in Method 2 with the seed set $ C_{max}(\Sigma,S_{xu})\times D^{p}$.} In other words, this model predictive control formulation implicitly embeds the results of Method 2.} 

In this work, we want to study the general properties of controlled invariant sets of $\Sigma_{p}$. For instance, is a longer preview always a better choice? How does the maximal controlled invariant set change as the preview time $p$ increases? Then, we study a special class of systems where the closed-form expression of the maximal controlled invariant set of the $p$-augmented systems can be derived analytically.

\section{Analytical Results}
In this section, we present analytical inner and outer approximations of controlled invariant sets for systems with different preview times. We also provide examples where the approximations are tight or not tight. Moreover, based on the approximations, we discuss {strategies to choose the preview time $p$}. An intuitive strategy is to select $p$ as large as possible, since a longer preview time provides more information than a shorter preview. However, since the dimension of $ \Sigma_p $ is proportional to $ p $, the existing methods suffer from the curse of dimensionality if the preview time is too long. Thus, we need a good strategy to select $p$, balancing between the computational cost and the performance.

First, the following theorem allows us to compare controlled invariant sets for systems with different preview times.
\begin{theorem} \label{thm:cartesian} 
	Suppose a set $C_{p_1} \subseteq \R^{n}\times D^{p_1}$ is a controlled invariant set of $\Sigma_{p_1}$ within $S_{xu,p_1}$ for some $p_1 \geq 0$. Then, for $p_2 > p_1$, $C_{p_1}\times D^{p_2-p_1}$ is a controlled invariant set of $\Sigma_{p_2}$ within $S_{xu,p_2}$.
\end{theorem}

Suppose  $p_2 > p_1 \geq 0$.  Thanks to Theorem \ref{thm:cartesian}, improvement {in safety control by} increasing the preview time from $ p_1 $ to $ p_2 $ can be measured by the volume difference of $ C_{max,p_2}$ and $ C_{max,p_1}\times D^{p_2-p_1} $. Moreover, $C_{max,p_1}$ provides an inner bound for $C_{max,p_2}$, that is 
\begin{align} \label{eqn:lifting} 
C_{max,p_1} \times D^{p_{2}-p_1} \subseteq  C_{max,p_2}. 
\end{align}
As a result, for all states $(x, d_{1:p_2})\in S_{xu,p_2}$, 
\begin{align}
\mathcal{A}((x, d_{1:p_1}), C_{max,p_1}) \subseteq \mathcal{A}((x, d_{1:p_2}),C_{max,p_2}).
\end{align}
That is, the maximal admissible input set at each state grows as the preview time increases.
An important question is then if there exists a critical $p_0$ such that the maximal admissible input set stops growing for $p > p_0$, that is for all $p>p_0$, for all states $(x, d_{1:p})\in S_{xu,p}$, 
\begin{align}
\mathcal{A}((x, d_{1:p_0}), C_{max,p_0}) =\mathcal{A}((x,d_{1:p}),C_{max,p}).
\end{align}
If such a $p_0$ indeed exists, we know the longest preview time to be considered is $p_0$, since preview longer than $p_0$ does not provide more admissible inputs.  In the next section, we show that this $p_0$ does exist for a specific class of systems. However, for general systems,  $p_0$ may not exist, shown by the following example.

\begin{example} \label{exp:1d}     
	Consider a $1$-dimensional system $$\Sigma: x(t+1) = ax(t) + u(t) + d(t),$$ with $x(t)$, $u(t)\in \R$ and $d(t)\in [- \gamma, \gamma]$. The safe set $S_{xu} = [-r, r]\times [- \beta, \beta]$. 
	
	Suppose that the parameters $ a$, $ \gamma$, $ \beta$ and $p$ satisfy $a > 1$, $r \geq (\beta+ \gamma)/(a-1)$ and $ a^{p-1}\beta \geq  \gamma$. Then, the maximal controlled invariant set $C_{max,p}$ of the $p$-augmented system within the augmented safe set $[-r,r]\times [- \gamma, \gamma]^{p}\times [- \beta, \beta]$ is the set of points $(x,d_1, \cdots, d_{p})$ satisfying\footnote{The proof can be found in Appendix.}
	
	\begin{itemize}
		\item[] (i) $d_{i} \in [- \gamma, \gamma]$ for $i$ from $1$ to $p$,
		\item[] (ii) $\vert x + \sum^{p}_{i=1} d_{i}/ a^{i} \vert \leq \frac{\beta - \gamma/ a^{p}}{(a-1)}$. 
	\end{itemize}
	Based on the closed-form expression of $C_{max,p}$, we can easily verify that for all $p\geq 0$, the maximal controlled invariant set $C_{max,p+1}$ strictly contains $C_{max,p}\times [- \gamma, \gamma]$ and thus the maximal admissible input set $ \mathcal{A} ((x, d_{1:p+1}), C_{max,p+1} ) $ strictly contains $  \mathcal{A}((x,d_{1:p}), C_{max,p})$ for some $(x, d_{1:p+1})\in C_{max,p+1}$. \END
\end{example}

{Example \ref{exp:1d} reveals that the maximal controlled invariant set may not converge at finite $p_0$ in the sense of $C_{max,p} = C_{max,p_0}\times D^{p-p_0}$ for $p \geq p_0$. } Then, to understand the asymptotic properties of $C_{max,p}$ as $p$ goes to infinity, we consider the {disturbance-collaborative system of $\Sigma$}:
\begin{align} \label{sys:inf} 
\mathcal{D}(\Sigma): x(t+1) = f(x(t),u(t),u_{d}(t)) 
\end{align}
with $x(t)\in \R^{n}$, $u(t)\in R^{m}$ and $u_{d}(t)\in R^{l}$. $A, B$ matrices are the same as in $\Sigma$. $u_{d} $ and $u$ are both input signals of $\mathcal{D}(\Sigma)$. The safe set of $\mathcal{D}(\Sigma)$  on $(x,u,u_{d})$ is $S_{xu, co} = S_{xu}\times D $. 

We denote the maximal controlled invariant set $C_{max}(\mathcal{D}(\Sigma), S_{xu,co})$ by $C_{max,co}( \Sigma, S_{xu})$, or $ C_{max,co} $ when $\Sigma$ and $S_{xu}$ are clear from the context. Intuitively, $ C_{max,co} $ contains all the possible initial states $ x $ from which the future state-input pairs of $ \Sigma $ can stay in $ S_{xu} $ indefinitely, when we have infinite preview time.

\begin{theorem} \label{thm:inf_p} 
	{The maximal controlled invariant set $ C_{max,p} $ of $ \Sigma_p $ within $ S_{xu,p} $ is a subset of the Cartesian product of the maximal controlled invariant set $ C_{max,co} $ and the set $ D^p $}, that is $C_{max,p} \subseteq C_{max,co} \times D^{p}$.
\end{theorem}
By Theorem \ref{thm:inf_p}, we know that $PROJ_{1:n}(C_{max,p}) \subseteq C_{max,co}$. Actually $C_{max,co}$ is a tight bound on $PROJ_{1:n}(C_{max,p})$ in general, shown by the following example where the Hausdorff distance between $PROJ_{1:n}(C_{max,p})$ and $C_{max,co}$ converges to $0$.
\begin{example}
	We consider the same dynamics and safe set in Example \ref{exp:1d}. The projection of the maximal controlled invariant set onto the first coordinate is 
	\begin{align}
	PROJ_{1}(C_{max,p}) = [ - \frac{\beta+\gamma- 2 \gamma/ a^{p} }{a-1}, \frac{\beta+\gamma- 2 \gamma/ a^{p} }{a-1}].
	\end{align}
	The corresponding disturbance-collaborative system is $$x(t+1) = ax(t) + u$$ with the safe set $[-r,r]\times [- \beta- \gamma, \beta + \gamma] $.
	It is easy to check that the maximal controlled invariant set $C_{max,co} $ of the disturbance-collaborative system is $ [- ( \beta+ \gamma) / (a-1),  ( \beta + \gamma) / (a-1)].$  Thus, $PROJ_{1} (C_{max,p})$ is strictly contained by $C_{max,co}$ for all\footnote{Recall that in Example \ref{exp:1d} we assume that $a^{p-1} \beta \geq  \gamma$, which implies $p \geq 1+ (\log(\gamma) - \log ( \beta))/\log(a) $.} $p \geq 1+ (\log(\gamma) - \log ( \beta))/\log(a) $ and as $p$ goes to infinity, $PROJ_{1}(C_{max,p})$ converges to the interior of $C_{max,co}$, that is $$ \lim_{p \to \infty} PROJ_{1}(C_{max,p}) =Int(C_{max,co}) = ( - \frac{ \beta+ \gamma}{a-1},  \frac{ \beta + \gamma}{a-1} ). $$ \END 
\end{example}
However, the Hausdorff distance between the projection $PROJ_{1:n}(C_{max,p})$ and $C_{max,co}$ does not always converge to $0$ as $p$ goes to infinity, shown by the following example.
\begin{example}
	Consider system $x(t+1) = u(t) + d(t)$, with $x(t)$, $u(t)\in \R$ and $d(t)\in [-5,5]$. Suppose the safe set $S_{xu}=[-1,1]\times [-1,1]$.  Obviously, $C_{max,co} = [-1,1]$ but $C_{max,p} = \emptyset$ for all $p \geq 0$. \END
\end{example}
Combining Theorems \ref{thm:cartesian} and \ref{thm:inf_p} , given any $p$, the maximal controlled invariant set $C_{max,p}$ of $\Sigma_{p}$ within $S_{xu,p}$ is bounded by 
\begin{align}
C_{max,p'}\times D^{p-p'} \subseteq C_{max,p} \subseteq C_{max,co}\times D^{p}, \label{eqn:upper_lower_bounds} 
\end{align}
where $C_{max,p'}$ is the maximal controlled invariant set of $\Sigma_{p'}$ within $S_{xu,p'}$, for some $p' < p$.   The  cost of computing $C_{max,p'}$ and $C_{max,co}$ is independent of the preview time $p$, but the cost to compute $C_{max,p'}$ rises as $p'$ increases. An inner bound tighter than the left hand side in \eqref{eqn:upper_lower_bounds}  can be obtained by growing $C_{max,p'}\times D^{p-p'}$ via Method 2 with $X_{0,p}= C_{max,p'}\times D^{p-p'}$, which requires more computational cost.

In practice, according to \eqref{eqn:upper_lower_bounds}, if we already compute $C_{max,p'}$ for some $p'$ and wonder if it is worth taking more cost to compute $C_{max,p}$ for $p$ larger than $p'$, a useful strategy is to compare the volumes of $C_{max,p'}\times D^{p-p'}$ and $C_{max,co}\times D^{p}$. The volume difference of the two sets indicates what we can gain at most by further increasing preview time.

Another significant implication of \eqref{eqn:upper_lower_bounds} is that for any initial state not in $C_{max,co}$, the future state-input trajectory of the system $ \Sigma $ cannot stay within $ S_{xu} $ indefinitely no matter how long the preview time $p$ is. In other words, $C_{max,co}$ shows the limits of safety control with preview in terms of the allowable initial states. 

\section{Systems in Brunovsky canonical form with hyperbox safe sets} \label{sec:brunov} 
In this section, we study systems in Brunovsky canonical form with a single input\footnote{{The results in this section apply to multiple-input case, since in Brunovsky canonical form, a system with multiple inputs can be decoupled into several systems with single input\cite{antsaklis2006linear}.}}. Due to the simple structure of the systems in Brunovsky canonical form, we can derive a closed-form expression of the maximal controlled invariant set within hyperbox safe sets. Next, based on the closed-form expression, we show convergence properties of the maximal controlled invariant set as the preview time increases. In terms of generality, any controllable system can be converted to a system in Brunovsky canonical form via an invertible transformation (see \cite{antsaklis2006linear}), and thus our results on systems in Brunovsky canonical form is also useful for controllable systems.

The dynamics of a system $\Sigma_{B}$ in Brunovsky canonical form is
\begin{align} \label{eqn:sys_B} 
\Sigma_{B}: x(t+1) = \overline{A} x(t) + \overline{B} u(t) + d(t),
\end{align}
where  $x(t)\in \R^{n}$, $u(t)\in \R$, $d(t)\in D \subseteq \R^n$, and 
\begin{align} \label{eqn:A_b} 
\overline{A} =  \begin{bmatrix}
\mathbf{0}_{(n-1)\times 1} & \mathbf{I}_{n-1}\\
0 & \mathbf{0}_{1\times (n-1)}
\end{bmatrix},
\overline{B} =  \begin{bmatrix}
\mathbf{0}_{(n-1)\times 1} \\ 1	
\end{bmatrix}. 
\end{align}
The $ \mathbf{I}_{k}$ and $ \mathbf{0}_{j\times k}$ in \eqref{eqn:A_b} represent the identity matrix in $\R^{k\times k}$ and the matrix with all zero entries in $\R^{j\times k}$.

Suppose that $D$ is a polytope in $\R^{n}$, and $B_{d} = \Pi_{k=1}^{n}[c_{k,1}, c_{k,2}]$ is the smallest hyperbox containing $D$.  We consider a  hyperbox safe set $\mathbf{B}\times \R$, where the state $x$ is constrained within hyperbox $ \mathbf{B}= \Pi_{k=1}^{n}[b_{k,1}, b_{k,2}]$ and the input $u$ is unconstrained.  Denote the $p$-augmented system corresponding to $\Sigma_{B}$ by $\Sigma_{B,p}$. The $p$-augmented safe set is $\mathbf{B}\times D^{p}\times \R$. 

We first derive a necessary condition for the existence of nonempty controlled invariant sets of $\Sigma_{ \mathbf{B},p}$ within $ \mathbf{B}\times D^{p}\times \R$. The idea is based on the following observation: Given an input $u(t)$ at time $t\geq 0$, due to the special structure of $ \overline{A}$ and $\overline{B}$, the $ (n-k+1) $ th entry  $x_{n-k+1}(t+k)$ of the state at time $ t+k $ for $k$ with $1\leq k\leq n$ can be exactly expressed as
\begin{align}
\begin{split} \label{eqn:u_pg} 
x_{n-k+1}(t+k) &= u(t) +  \sum^{k-1}_{i=0}  d_{1,n-i}(t+i), 
\end{split}
\end{align}
where $d_{1,n-i}(t+i)$ is the $n-i$ th entry of  $d_1(t+i)\in \R^n$ for $i$ from $0$ to $n-1$. 

Suppose there exists a nonempty controlled invariant set in $ \mathbf{B}\times D^{p}\times \R$. Then, there exists at least one safe input $u(t)\in \R$ such that for all $k$ from $1$ to $n$, the right hand side of \eqref{eqn:u_pg} satisfies the constraints on $x_{n-k+1}(t+k)$ from $ \mathbf{B}$, robust to all possible future disturbances, that is, for $k$ from $1$ to $n$, 
\begin{align}
u(t) + \sum^{k-1}_{i=0} d_{1,n-i}(t+i)\in [b_{n-k+1,1}, b_{n-k+1,2}], \label{eqn:u_contain} 
\end{align}
for all possible values of $\sum^{k-1}_{i=0}d_{1,n-i}(t+i)$; otherwise, for all $u(t)\in \R$, we can find future disturbances such that the state $x_{n-k+1}(t+k)\not\in [b_{n-k+1,1},b_{n-k+1,2}]$.

Note that if $ i < p$,  $d_{1,n-i}(t+i)$ is a scalar known from preview at time $t$; otherwise $d_{1,n-i}(t+i)$ takes arbitrary values in $[c_{n-i,1}, c_{n-i,2}]$. Based on this observation, the condition of the existence of a safe input $u(t)$ satisfying  \eqref{eqn:u_contain} is given in Theorem \ref{thm:nec_p}, which is necessary for the existence of a nonempty controlled  invariant set. 

\begin{theorem} \label{thm:nec_p} 
	There exists a nonempty controlled invariant set of $\Sigma_{B,p}$ within $ \mathbf{B}\times D^{p}\times \R$ only if $\forall v\in V_{d, \overline{p}}$ we have
	\begin{align} \label{eqn:nec_cond_p} 
	\bigcap_{k=1}^{n} \left(\left[\widehat{b}_{k,1}, \widehat{b}_{k,2} \right]  - \sum^{\min(n-k+1, \overline{p})}_{i=1} v_{ \overline{p}-i+1}  \right)  \not= \emptyset
	\end{align}
	where $V_{d,\overline{p}}$ is the  set of vertices of the hyperbox $\mathbf{B}_{d,\overline{p}} = \Pi_{k=n- \overline{p}+1}^{n}[c_{k,1},c_{k,2}]$, and $ \overline{p} = \min (p, n)$, and $\widehat{b}_{k,1} = b_{k,1} -\sum^{n- \overline{p}}_{i=k}c_{i,1}$ and $ \widehat{b}_{k,2} = b_{k,2} - \sum^{n- \overline{p}}_{i=k}c_{i,2}$ for $k $ with $ n- \overline{p} \leq k \leq n$. 
\end{theorem}
In practice, if we want to compute controlled invariant sets of $\Sigma_{B,p}$, unnecessary computations can be avoided by checking the condition in \eqref{eqn:nec_cond_p} first. As the number of constraints in \eqref{eqn:nec_cond_p} is proportional to the cardinality of $ V_{d, \overline{p}} $, we derive an equivalent condition to \eqref{eqn:nec_cond_p} that contains only  $n^{2}$ inequalities: for all $j $ and $k$ from $1$ to $n$, 
\begin{align}  \label{eqn:nonempty} 
\begin{split}
&\forall j=k,\; b_{j,1}-b_{k,2} \leq \sum^{n-p}_{i=k}(c_{i,1}-c_{i,2}), \\
&\forall j<k,\; b_{j,1}-b_{k,2} \leq \sum^{n-p}_{i=j}c_{i,1}- \sum^{n-p}_{i=k} c_{i,2} + \sum^{\max(k-1,n-p)}_{i=\max(j,n-p+1)} c_{i,1}, \\
&\forall j>k,\; b_{j,1}- b_{k,2} \leq \sum^{n-p}_{i=j}c_{i,1}- \sum^{n-p}_{i=k} c_{i,2}- \sum^{\max(j-1,n-p)}_{i=\max(k,n-p+1)} c_{i,2}.
\end{split}
\end{align}

Next, suppose that there exists a nonempty controlled invariant set, namely that \eqref{eqn:nec_cond_p} is satisfied. We  derive conditions under which states $\xi = (x,d_1,d_2,...,d_{p})\in \R^{(p+1)n}$ are contained  by the maximal controlled invariant set.

We use $x_{i}$, $d_{k,i}$ to denote $i $ th entry of $x$, $d_{k}$. According to the dynamics in \eqref{eqn:sys_B}, the first $(n-t)$ entries of the vector $x(t)$ for all $t = 0, 1, \cdots  , n-1$ are independent from the control inputs and completely determined by the initial state  $x(0)$ and disturbances $d(0) $, $d(1)$ ..., $d(n-2)$. 

Thus, one necessary condition on $\xi(0) = (x(0), d_{1:p}(0)) \in C_{max,p}$  is that for all possible future disturbances in $ D $ that are not previewed yet at the initial time, for all $t$ from $ 0 $ to $ n-1 $ and all $k$ from $ 1 $ to $ n-t $,  the state $x(t)$ satisfies 
\begin{align} \label{eqn:state_nec_cond} 
\begin{split} 
x_{k}(t) \in [b_{k,1}, b_{k,2}].
\end{split} 
\end{align}
By expanding $x_{k}(t)$ using $x(0)$ and $ d_{1:p}(0) $, we obtain the conditions stated in the following theorem.
\begin{theorem} \label{thm:state_nec_cond} 
	A state	$(x,d_{1:p}) $ is contained in the maximal controlled invariant set $ C_{max,p}$ only if
	\begin{align}
	\label{eqn:C_p_1} 
	&x\in \mathbf{B}, d_{1:p}\in D^{p},
	\end{align}
	and  for all $ k $, $ 2 \leq k \leq n$ and for all $j $, $ 1 \leq j < k:$
	\begin{align}
	x_{k} + \sum^{\min(k-j,p)}_{i=1} d_{i,k-i} \in		[b_{j,1}, b_{j,2}] - \sum^{k-j}_{i=p+1} [c_{k-i,1},c_{k-i,2}]. \label{eqn:C_p_2} 
	\end{align}
	where $d_{i,k-i}$ is the $k-i$ th entry of vector $d_{i}$. 
\end{theorem}
{To clarify the notation, in the case of $k-j < p+1$, the right hand set of \eqref{eqn:C_p_2} becomes $[b_{j,1}, b_{j,2}] - \emptyset = [b_{j,1}, b_{j,2}]$.}
We denote the set of states $ (x,d_{1:p})$ satisfying constraints in  \eqref{eqn:C_p_1} and \eqref{eqn:C_p_2} by $C_{p}$. The following theorem states that the maximal controlled invariant set of $\Sigma_{B,p}$ within $ \mathbf{B}\times D^{p}\times \R$ is exactly equal to $C_{p}$.
\begin{theorem} \label{thm:max_inv_set} 
	Suppose that \eqref{eqn:nec_cond_p}  is satisfied. Define
	\begin{align}
	C_{p} = \{ \xi = (x,d_{1:p}) \mid \xi \text{ satisfies \eqref{eqn:C_p_1}, \eqref{eqn:C_p_2}}  \}. \label{eqn:C_p} 
	\end{align} 
	Then, $C_{p}$ is the maximal controlled invariant set of $\Sigma_{B,p}$ within the safe set $ \mathbf{B}\times D^{p}\times \R$.
\end{theorem}

\begin{corollary} \label{cor:nonempty} 
	The condition in \eqref{eqn:nec_cond_p} 
	is necessary and sufficient for the existence of nonempty controlled invariant sets of $\Sigma_{B,p}$ within $ \mathbf{B}\times D^{p}\times \R$.
\end{corollary}

\begin{corollary} \label{cor:variant} 
	If instead of $\Sigma_{B}$ in \eqref{eqn:sys_B}, we consider a system in the following form:
	\begin{align} \label{eqn:sys_B_var} 
	\Sigma_{v}: x(t+1) = \overline{A} x(t)+ \overline{B}u(t) + \overline{E} d(t),
	\end{align}
	for $d(t)\in D_{v} \subseteq \R^{l}$ and some $\overline{E}\in \R^{n\times l}$. Then, we first define system $\Sigma_{B}'$ in Brunovsky canonical form  
	\begin{align}
	\Sigma_{B}': x(t+1) = \overline{A}x(t) + \overline{B}u(t) + d(t),
	\end{align}
	with $d(t)\in \overline{E}D_{v} \subseteq \R^{n}$. We have the closed-form expression of the maximal controlled invariant set $C_{p}$ of the $p$-augmented system of $\Sigma_{B}'$ within $  \mathbf{B}\times D^{p}\times \R$. The maximal controlled invariant set $C_{v}$ of the $p$-augmented system of $\Sigma_{v}$ within $ \mathbf{B}\times D^{p}\times \R$ is nonempty if and only if $C_{p}$ is nonempty and
	\begin{align}
	C_{v} = \{(x,  d_{1:p}) \mid (x, \overline{E}d_1, \overline{E}d_{2}, \cdots, \overline{E}d_{p})\in C_{p}\}.  \label{eqn:43} 
	\end{align}
\end{corollary}

\begin{remark} \label{rm:two_move} 
	Adopting the idea from \cite{tzanis2021acc}, for a more general safe set in form of $P\times \R$, where $P$ is a polytope, we can construct a controlled invariant set of $\Sigma_{B,p}$	within $P\times D^{p}\times \R$ in $2$ moves: First, we construct a polytope in a lifted space that encodes all hyperboxes $ \mathbf{B}$ in $P$ and all states $(x,d_{1:p})$ within the maximal controlled invariant set within $ \mathbf{B}\times D^{p}\times \R$, based on the nonemptyness condition \eqref{eqn:nec_cond_p} and the closed-form expression of $C_{p}$. Then, we project this lifted set onto its first $n(p+1)$ coodinates, equal to the union of the maximal controlled invariant set within $ \mathbf{B}\times D^{p}\times \R$ for all hyperboxes $ \mathbf{B}$ contained by $P$. By construction, this set is a controlled invariant set in $P\times D^{p}\times \R$. 
	
	Furthermore, as stated in Remark 1 of \cite{anevlavix2019computing}, any controllable system with a polytopic safe set (including input constraints) can be transformed into system in Brunovsky canonical form with a safe set in form of $P\times \R$. Thus, our results in this section can be used to compute controlled  invariant sets for $p$-augmented systems of a controllable system.
	\END
\end{remark}

According to the closed-form expression of the maximal controlled invariant set $C_{p}$, we show the convergence property of $C_{p}$ for $p \geq  n$ in the following theorem.
\begin{theorem} \label{thm:p_g_n} 
	For preview time $p > n$, the maximal controlled invariant set $C_{p}$ is equal to the Cartesian product of the maximal controlled invariant set $C_{n}$ of $\Sigma_{B,n}$ and the set $D^{p-n}$, that is $C_{p} = C_{n}\times D^{p-n}$.
\end{theorem}

Theorem \ref{thm:p_g_n} indicates that for system $\Sigma_{B}$ in Brunovsky canonical form with a safe set $\mathbf{B}\times \R$, the preview time longer than $p= n$ is not necessary. However, given a state $(x,d_{1:p})$ in the maximal controlled invariant set $C_{n}\times D^{p-n}$, the admissible input set with the maximal size is obtained when preview is $n+1$, that is
\begin{align*}
\mathcal{A}(C_{n}, (x, d_{1:n})) \subseteq \mathcal{A}(C_{n+1}, (x, d_{1:n+1})) 
=  \mathcal{A}( C_{p}, (x,d_{1:p})).
\end{align*}
That is, the critical preview time $p_0  = n+1$.

We are curious if the property $p_0=n+1$ holds for systems in Brunovsky canonical form with arbitrary polytopic safe sets. Unfortunately, the following example shows that for general safe sets, a critical preview time may not exist.  
\begin{example}
	Consider the $1$-dimensional system $\Sigma$ and the safe set $S_{xu}$ defined in Example \ref{exp:1d}. We replace $u(t)$ in $\Sigma$ by $u(t) = -ax(t) + v(t),$ where $v(t)$ is the new control input. Then, the $1$-dimensional dynamics $\Sigma'$ with respect to the state $x$ and the input $v$ is in Brunovsky canonical form. The safe set for this new dynamics is $S_{xu}' = \{(x,v) \mid (x,-ax + v)\in S_{xu}\} $.  
	
	Let $C_{max,p}$ and $C_{max,p}'$ be the maximal controlled invariant sets of $\Sigma$ within $S_{xu}$ and $\Sigma'$ within $S_{xu}'$ respectively. It can be easily shown that $C_{max,p}' = C_{max,p}$. Thus, $C_{max,p}'$ strictly contains $C_{max,n}'\times [- \gamma, \gamma]^{p-n}$.\END
\end{example}

Finally, recall that an outer bound on controlled invariant sets of $\Sigma_{B,p}$ is given in Section \ref{sec:prelim} by the Cartesian product of {the maximal controlled invariant set of the disturbance-collaborative system and the set $D^{p}$, that is the right hand set of \eqref{eqn:upper_lower_bounds}. We wonder the relation between $C_{n}$ and this outer bound, which is revealed by the next theorem.
	
	\begin{theorem} \label{thm:two_bounds} 
		For preview time $p>n$, if nonemptyness condition \eqref{eqn:nec_cond_p} holds, then the projection of $C_{p}$ onto the first $n$ coordinates is equal to the maximal controlled invariant set $C_{max,co}$ of the disturbance-collaborative system $\mathcal{D}(\Sigma_{B})$ within safe set $\mathbf{B}\times \R$, that is $$C_{max,co} =PROJ_{1:n}(C_{p}) = PROJ_{1:n}(C_{n}).$$	
	\end{theorem}

	\section{Illustrative Examples}
	In this section, we want to study the benefits of preview on disturbances via several concrete examples. 
	
	\subsection{Impact of Preview on Disturbance Tolerance}
	We demonstrate the impact of preview on disturbance tolerance via our results on systems in Brunovksy canonical form. We fix the state dimension $n=10$ and the safe set $ \mathbf{B} = \Pi_{i=1}^{n} [-1,1]$. Then, we parametrize the disturbance set $D = \Pi_{i=1}^{n} [-c,c]$ by a positive number $c>0$. We are interested in the largest $c$ for the augmented system $\Sigma_{B,p}$ to have nonempty controlled invariant sets within $ \mathbf{B}\times D^{p}\times \R$. According to Corollary \ref{cor:nonempty}, we can utilize the condition on nonempty controlled invariant set given by \eqref{eqn:nonempty} to determine the largest possible $c$. 
	
	By plugging $b_{k,1} = -1$, $b_{k,2}=1$, $c_{k,1} = -c$ and $c_{k,2} = c$ for all $k$ from $1$ to $n$ into  \eqref{eqn:nonempty} , we obtain an upper bound on $c$ such that \eqref{eqn:nonempty}  holds. The largest $c$ computed for different preview time $p$ are shown in Fig. \ref{fig:largest_c}. As we expect, when the preview time increases, a larger disturbance set can be handled, due to the power of preview.  
	
	\begin{figure}[]
		\centering
		\includegraphics[width=0.45\textwidth,trim=0 0 0 -10]{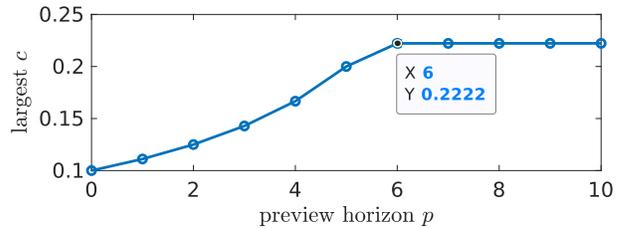}
		\caption{\small The largest disturbance bound $c$ versus preview time $p$ for the system in Brunovsky canonical form ($n=10$) with hyperbox safe set. }
		\label{fig:largest_c}
		\vspace{-0.5cm}
	\end{figure}
	
	In addition, we observe in Fig. \ref{fig:largest_c} that the largest $c$ stops increasing after $p \geq 6$. This observation suggests that a disturbance set with $c > 0.2222$ may lead to an empty controlled invariant set for any preview time $p$. With some calculation, it can be verified that for $c > 2/9$, the necessary condition \eqref{eqn:nonempty} does not hold for all $ p \geq 0$ and thus the maximal controlled invariant set is always empty no matter how large the $p$ is. 
	
	\subsection{Lane Keeping Control with Preview}
	To show the usefulness of preview, we present how preview helps the driver-assist system to keep a vehicle within lanes. We use a $4$-dimensional linearized bicycle model with respect to constant longitudinal speed $30m/s$ from \cite{smith}. The state space consists of lateral displacement $y$, lateral velocity  $v$, yaw angle $\Delta \Psi$ and yaw rate $r$. The disturbance $r_{d}$ with $ \vert r_{d} \vert \leq 0.04$ considered in this simplified model is a quantity related to the road curvature that perturbs the yaw angle. The control input $u$ is the steering angle, with constraints $u\in [-\pi/2,\pi/2]$. 
	
	The safe set $S_{xu}$ is the set of state-input pairs within bounds $ \vert y \vert\leq 0.9 $, $ \vert v \vert \leq 1.2$, $ \vert \Delta\Phi \vert \leq 0.05 $ and $ \vert r \vert \leq 0.3 $, and $ \vert u \vert \leq \pi/2$. We set the preview time $p=5$. We first compute the maximal controlled invariant set within $S_{xu}$ for system without preview, denoted by $C_{max,0}$. Then, we use Method 2 to grow the seed set $C_{max,0}\times D^{5}$ for the $p$-augmented system over $10$ iterations, the result of which is denoted by $C_{io,5}$. Numerically we find that $C_{io,5}$ strictly contains $C_{max,0}\times D^{5}$. We also try the idea in Remark \ref{rm:two_move} to obtain a controlled invariant set based on our results in Section \ref{sec:brunov}, but the resulting set is contained by $C_{max,0}\times D^{p}$, which is too conservative to be useful.
	
	Next, we find a point $(x_0, d_1, \cdots, d_{5}) $ belonging to the set difference $C_{io,5} \setminus C_{max,0}\times D^{5}$ and simulate $2$ trajectories starting at $x_0$ with the first $5$ disturbances $d_{1:5}$, using the two controlled invariant sets $C_{max,0}$ and $C_{io,5}$ respectively. The controller consists of $2$ parts: First, we have a nominal state feedback controller, designed via linear quadratic regulator for the $ p $-augmented system. Then, at each time instant, we supervise the control input from the nominal controller by projecting that input onto the admissible input set at current state with respect to $C_{max,0}$ or $C_{io,5}$. If the admissible input set happens to be empty at some time instants, then we project the nominal input onto the input constraint set $[-\pi/2,\pi/2]$. The resulting vehicle maneuvers are shown by Fig. \ref{fig:maneuver}, where we find that the trajectory under the supervision of the admissible input set with respect to $C_{io,5}$ stays within the lane as required by the safety constraints during the simulation time span, but the trajectory under the supervision with respect to $C_{max,0}$ violates the constraints on lateral displacement $y$ and drives out of the lane at the $2$nd time step. This observation meets our expectation since the initial condition was not in $C_{max,0}$. This example demonstrates how the preview on future disturbances enables controllers to operate safely from a larger set of initial conditions.
	
	\begin{figure}[]
		\centering
		\includegraphics[width=0.36\textwidth, trim=0 0 0 -10]{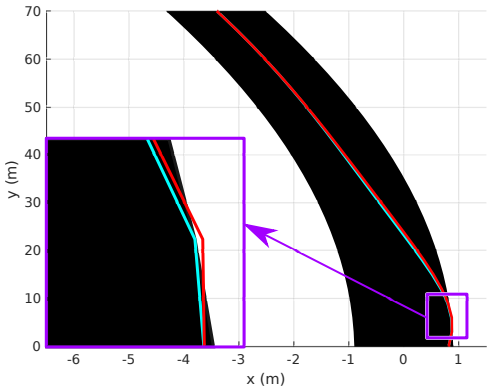}
		\caption{\small The vehicle maneuvers under the linearized bicycle model with supervised LQR controller. The dark region indicates the safe region in the plane, that is the lane. The cyan curve is the maneuver corresponding to  $C_{oi,5}$. The red curve is the maneuver corresponding to $C_{max,0}$. }
		\label{fig:maneuver}
		\vspace{-0.7cm}
	\end{figure}
	
	\section{Conclusion}
	In the first part of this work, we study general properties of controlled invariant sets for systems with preview and the implications of those properties, including a strategy to choose a preview time. In the second part, we study systems in Brunovsky canonical form with hyperbox safe sets, for which we derive the maximal controlled invariant set of the $ p $-augmented system in closed form. The impact of preview on the controlled invariant sets can be directly analyzed using this closed-form expression, by help of which we prove the existence of a critical preview time for this class of systems. In future work, we plan to study noisy preview information.

\bibliographystyle{IEEEtran}
\bibliography{ref}
\section{Appendix}

\begin{proof}[Proof of Theorem \ref{thm:cartesian}]
	Let  $(x,d_{1:p_2}) \in C_{p_1}\times D^{p_2-p_1}$. We want to prove that there exists a safe input $u_s$ such that $(x,u_{s})\in S_{xu}$ and $(f(x,u_{s},d_1), d_{2:p_2}, d)\in C_{p_1}\times D^{p_2-p_1}$ for all $d\in D$. That is, find a $u_{s}$ such that $(x,u_{s})\in S_{xu}$ and $((f(x,u_{s},d_1),d_{2:p_{1}+1})\in C_{p_1}$. Such a $u_{s}$ can be picked from the admissible input set $\mathcal{A}((x,d_{1:p_1}), C_{p_1})$ at the truncated state $(x,d_{1:p_1})$.
\end{proof}

\begin{proof}[Proof of Example \ref{exp:1d}]
	Suppose we know $C_{max,p}$. Then compute the set 
	\begin{align}
	F = \{a^{p} x + \sum^{p}_{i=1} a^{p-i}d_{i} \mid (x, d_{1:p})\in C_{max,p}\}. 
	\end{align}
	Consider the auxiliary dynamics 
	\begin{align}
	\Sigma_{aux}: z(t+1) = az(t) + a^{p} u(t) + d(t),
	\end{align}
	with $d(t)\in [- \gamma, \gamma]$. We want to show $F$ is a controlled invariant set of the auxiliary system within the  auxiliary safe set $\R\times [- \beta, \beta]$.
	
	Let $z\in F$. There exists $(x,d_{1:p})\in C_{max,p}$ such that $z = a^{p}x + \sum_{i=1}^{p}a^{p-i}d_{i}$. There exists $u\in [- \beta, \beta] $ such that $(ax+u+d_1,d_{2}, \cdots, d_{p}, d)\in C_{max,p}$ for all $d\in [- \gamma, \gamma]$. Then, $ a z +  a^{p}u + d = a^{p+1} x + a^{p}u +a^{p}d_1+ \sum^{p}_{i=2} a^{p-i+1} d_i +d \in F$  for all $d\in [- \gamma, \gamma]$.  That is, $F$ is controlled invariant with respect to the auxiliary dynamics.
	
	Since the safe set, input set and disturbance set are symmetric,  $C_{max,p}$ must be a convex set symmetric with respect to origin, that is $\xi\in C_{max,p}$ if and only if $-\xi\in C_{max,p}$. Thus, $F$ has to be an interval in $\R$ symmetric with respect to $0$. Suppose $F = [-b,b]$. For $F$ being controlled invariant, $b$ needs to be greater than or equal to $\gamma$, and there needs to exists $u\in [- \beta, \beta]$ such that $a b + a^{p}u + d \leq b $ for all $d\in [- \gamma, \gamma]$, which implies $ ab  - a^{p} \beta + \gamma \leq b$. That is $ b\leq  (a^{p} \beta - \gamma) / (a-1): = \overline{b}$. Thus, we have $F \subseteq [ -\overline{b}, \overline{b}] $.  (Note that for $F$ being nonempty, parameters $a$,  $ \beta$ and $ \gamma$ must satisfy $  ( a^{p} \beta - \gamma) / (a-1) = \overline{b} \geq \gamma$, that is $a^{p-1} \beta \geq  \gamma$.) Thus,  $C_{max,p} $ must be contained by the set $ C:= \{(x,d_{1:p}) \mid \vert a^{p}x+ \sum^{p}_{i=1}a^{p-i}d_{i} \vert \leq \overline{b}, \vert d_{i} \vert \leq  \gamma  \}. $

	Next, we want to show that $C$ is a controlled invariant set within $[-r,r]\times [- \gamma, \gamma]^{p}\times [- \beta, \beta]$. Let $(x,d_{1:p})\in C$. Define $z= a^{p}x + \sum^{p}_{i=1} a^{p-i}d_{i} $. Then, $ \vert z \vert \leq  \overline{b}$ and thus $ \vert z/a^{p} \vert = \vert x + \sum^{p}_{i=1} d_{i} / a^{i}  \vert \leq  \overline{b}/ a^{p} =  ( \beta - \gamma/ a^{p})/ (a-1)$. Note that $ \vert d_{i} \vert \leq  \gamma $. Thus, $ \vert x \vert \leq  ( \beta - \gamma/ a^{p}) / (a-1) + \gamma\sum^{p}_{i=1}1/a^{i}  < (\beta + \gamma)/ (a-1) \leq r$. Thus, $C \subseteq [-r,r]\times [- \gamma, \gamma]^{p}$.  Also, it is easy to check that  $[ - \overline{b}, \overline{b}]$ is a controlled invariant set of the auxiliary system within $\R\times [- \beta, \beta]$. Thus, there exists $u\in [- \beta, \beta]$ such that $ \vert az + a^{p}u + d \vert \in [-\overline{b}, \overline{b}]$ and thus $(ax+ u + d_1, d_2, \cdots, d)\in C$ for all $d\in [- \gamma, \gamma]$. Thus, $C$ is a controlled invariant set within $[-r,r]\times [- \gamma, \gamma]^{p}\times [- \beta, \beta]$ and thus $ C = C_{max,p} $.   
\end{proof}

\begin{proof}[Proof of Theorem \ref{thm:inf_p}]
	Denote $C_{x,p} = PROJ_{1:n}( C_{max,p})$. We want to prove that $C_{x,p}$ is a controlled invariant set of $\mathcal{D}(\Sigma)$ within $S_{xu,co} = S_{xu}\times D$. 
	
	Let $x\in C_{x,p}$. Then, there exists $d_{1:p}\in D^{p}$ such that $(x,d_{1:p})\in C_{max,p}$. Since $C_{max,p}$ is a controlled invariant set $\Sigma_{p}$ within $S_{xu,p}$, there exists $u$ such that $(x, u)\in S_{xu}$, and also $(f(x,u,d_1), d_{2:p},d ) \in C_{max,p}$ for all $d\in D$, which implies $f(x,u,d_1)\in C_{x,p}$. Let $u_1 = u$ and $u_2 = d_1\in D$. Then, $(x, u_1,u_2)\in S_{xu}\times D$ and $x^{+} = f(x,u_1,u_2)\in C_{x,p}$. Thus, $C_{x,p}$ is controlled invariant with respect to $\Sigma_{d}$ within $S_{xu,co}$.  
	
We have	$C_{x,p} \subseteq C_{max,co}$, since $C_{max,co}$ is the maximal controlled invariant set within $S_{xu,co}$. By definition of projection, we have for all $p \geq 0$,
	\begin{align}
	C_{max,p} \subseteq C_{x,p} \times D^{p} \subseteq C_{max,co} \times D^{p}.
	\end{align}
\end{proof}

\begin{proof}[Proof of Theorem \ref{thm:nec_p}]
	As discussed in the paragraphs above Theorem \ref{thm:nec_p}, we want to derive necessary conditions on boundaries $b_{k,1}$ and $b_{k,2}$ of  $ \mathbf{B}$ such that \eqref{eqn:u_contain} holds. That is,  there exists $u(t)$ such that for $k$ from $1$ to $n$, 
	\begin{align}
		u(t) + \sum^{k-1}_{i=0} d_{1,n-i}(t+i)\in [b_{n-k+1}, b_{n-k+1,2}]. \label{eqn:thm_3_1} 
	\end{align}
	Suppose that the $p$-step preview at time $t$ is $d_{1:p}(t)$.  Then, for $i<p$, $d_{1,n-i}(t+i)=d_{i+1, n-i}(t)$. For $i \geq p$, $d_{1,n-i}$ is not previewed at time $t$ and thus can take arbitrary value in $[c_{n-i,1}, c_{n-i,2}]$. Thus, condition in \eqref{eqn:thm_3_1} is equivalent to the condition that there exists $u(t)$, for $k$ from $1$ to $n$, 
	\begin{align}
		u(t) + \sum^{\min (k,p)-1}_{i=0} d_{i+1,n-i}(t) + \sum^{k-1}_{i=p} [c_{n-i,1},c_{n-i,2}]\nonumber \\\subseteq [b_{n-k+1,1},b_{n-k+1,2}].  \label{eqn:thm_3_2} 
	\end{align}
By moving the second and third terms on the left side of \eqref{eqn:thm_3_2} to the right, \eqref{eqn:thm_3_2} becomes 
\begin{align}
		u(t)\in [\widehat{b}_{n-k+1,1},\widehat{b}_{n-k+1,2}] - \sum^{\min (k,p)-1}_{i=0} d_{i+1,n-i}(t), \label{thm_3_3} 
\end{align}
where $ \widehat{b}_{n-k+1,1}$ and $\widehat{b}_{n-k+1,2}$ are defined in Theorem \ref{thm:nec_p}. The necessary and sufficient condition of $u(t)$ satisfying \eqref{thm_3_3} for all $k$ from $1$ to $n$ is
\begin{align}
	\bigcap_{k=1}^{n} \left([\widehat{b}_{n-k+1,1},\widehat{b}_{n-k+1,2}] - \sum^{\min (k,p)-1}_{i=0} d_{i+1,n-i}(t)\right)\not=\emptyset. \label{eqn:thm_3_4} 
\end{align}
Finally, we denote $ \overline{p} = \min(n,p)$. Note that for $t$ longer than $p$, the vector $(d_{ \overline{p}, n- \overline{p}+1}(t), \cdots, d_{2,n-1}(t), d_{1,n}(t))$, denoted by $v$,  can take arbitrary value in $ \mathbf{B}_{d, \overline{p}} = \Pi_{k=n- \overline{p}}$.  Thus, for all $v\in \mathbf{B}_{d, \overline{p}}$, we need \eqref{eqn:thm_3_4} holds. That is equivalent to check that \eqref{eqn:thm_3_4} holds for all vertices of $ \mathbf{B}_{d, \overline{p}}$, which is the condition in \eqref{eqn:nec_cond_p} . 
\end{proof}

\begin{proof}[Proof of Theorem \ref{thm:state_nec_cond}]
 First, the condition in \eqref{eqn:C_p_1} is necessary since $C_{max,p} \subseteq B\times D^{p}$. Next, given that \eqref{eqn:C_p_1} holds, we want to derive conditions equivalent to the condition in \eqref{eqn:state_nec_cond}. 
	
	Let $x(0)= x $ and $ d(k)= d_{k+1}$ for $k$ from $0$ to $p-1$. 	According to dynamics \eqref{eqn:sys_B}, for any $t$, $ 0< t \leq  n-1$, for any $k$, $1 \leq k \leq n-t$, $x_{k}(t) = x_{k+t}(0) + \sum^{t}_{i=1} d_{k+t-i}(i-1) $. Define $\overline{p} = \min(p,t)$. Then, we can write $ x_{k}(t) =  x_{k+t}(0) + \delta_1 + \delta_2$, where $ \delta_1 = \sum^{ \overline{p}}_{i=1} d_{k+t-i}(i-1)$ and $ \delta_2 = \sum^{t}_{i= \overline{p}+1} d_{k+t-i}(i-1) = \sum^{t}_{i= p+1} d_{k+t-i}(i-1) $. Note that $ \delta_1$ is determined by the preview at time $t=0$, that is $d(t)$ for $t$ from $0$ to $p-1$. $ \delta_2$ is determined by future disturbances not previewed at $t=0$. For $t \geq  \overline{p}+1$, $d_{k+t-i}(i-1)$ can take arbitrary values in interval $ [c_{k+t-i,1},  c_{k+t-i,2}]$. Thus, $ \delta_2$ can take arbitrary values in interval $ \sum^{t}_{i= p+1} [c_{k+t-i,1}, c_{k+t-i,2}]$. Thus, depending on the value of $ \delta_2$, $x_{k}(t)$  can be any value in the interval $ x_{k+t}(0) + \delta_1 + \sum^{t}_{i= {p}+ 1} [c_{k+t-i,1},c_{k+t-i,2}]. $ To guarantee that $x_{k}(t) \in [b_{k,1}, b_{k,2}]$, it is necessary to have $$ x_{k+t}(0) + \delta_1 + \sum^{t}_{i={p}+1} [c_{k+t-i,1}, c_{k+t-i,2}]  \subseteq [b_{k,1}, b_{k,2}], $$ which is equivalent to
	\begin{align}  
	x_{k+t}(0) + \delta_1 \in [b_{k,1}, b_{k,2}]- \sum^{t}_{i= {p}+1} [c_{k+t-i,1},c_{k+t-i,2}]. \label{eqn:thm_4_1} 
	\end{align}
	By plugging the expression of $ \delta_1$ in the above formula, we obtain the condition in \eqref{eqn:C_p_2}. 
\end{proof}

\begin{proof}[Proof of Theorem \ref{thm:max_inv_set}]
	Since \eqref{eqn:C_p_1}  and \eqref{eqn:C_p_2} are necessary conditions for a state contained in $C_{max,p}$, we know that $C_{max,p} \subseteq C_{p}$. Next, we want to show that $C_{p}$ is a controlled invariant set of $\Sigma_{B,p}$ within $ \mathbf{B}\times D^{p}\times \R$, which implies $C_{p} = C_{max,p}$. 
	
	Let $(x(0),  d_{1:p}(0))\in C_{p}$. We want to find an input $u(0)\in \R$ such that the next state $ (x(1), d_{1:p}(1)) = (\overline{A}x(0)+\overline{B}u(0) + d_1(0), d_{2:p}(0), d(0))\in C_{p}$ for all $d(0)\in D$.  The idea for the remainder of the proof is to explicitly construct a $u(0)$ that satisfies the condition above. 

	Let $\overline{p} = \min(p,n)$. We define $v = (d_{\overline{p},n- \overline{p}+1}, d_{ \overline{p}-1,n- \overline{p}+2}, \cdots, d_{1,n})\in \R^{\overline{p}}$ where $d_{j,k}$ is the $k$ th entry of the vector $d_{j}(0)$. Then, the vector $v$ is contained by the convex hull $ CH(V_{d,\overline{p}})$ for $V_{d,\overline{p}}$ in Theorem \ref{thm:nec_p}. That is, there exists $ e_1$, ..., $e_{\overline{p}}\in V_{d, \overline{p}}$ and $ \alpha_{1}$, ..., $ \alpha_{ \overline{p}} \geq 0$ such that $\sum_{i} \alpha_{i} = 1$ and $v = \sum_{i=1} \alpha_{i} e_{i} $.
	For simplicity, given point $e_{i} \in V_{d,\overline{p}}$, we denote the interval on the left hand side of \eqref{eqn:nec_cond_p} corresponding to $e_{i}$ by $I(e_{i})$.
	For each $i \in \{1, \cdots, \overline{p}\} $, since \eqref{eqn:nec_cond_p} holds, there exists $u_{i}\in I(e_{i})$. Let $u(0) = \sum_{i=1} \alpha_{i} u_{i} $. 

	By construction of $u(0)$ and the proof of Theorem \ref{thm:nec_p}, $u(0)$ satisfies the condition \eqref{eqn:u_contain}  for any $d(0)\in D$, that is $x_{n-k+1}(k) \in [b_{n-k+1}, b_{n-k+1,2}]$ for all $k$ from $1$ to $n$ for any $d(0)\in D$.  Also, by the proof of Theorem \ref{thm:state_nec_cond},  $(x(0), d_{1:p}(0))\in C_{p}$ implies that $d_{1:p}(0) \in D^{p}$ and for $d(0)\in D$,  $x_{k}(t) \in [b_{k,1},b_{k,2}]$ for all $t$ from $0$ to $n-1$ and all $k$ from $1$ to $n-t$. Thus, for $(x(1), d_{1:p}(1))$ with respect to any $d(0)\in D$, we have $d_{1:p}(1)\in D^{p}$ and $x_{k}(t+1) \in [b_{k,1},b_{k,2}]$ for all $t$ from $0$ to  $n-1$ and all $k$ from $1$ to $n-t$, which implies $(x(1), d_{1:p}(1),d(0))\in C_{p}$ by the proof of Theorem \ref{thm:state_nec_cond}. 
\end{proof}

\begin{proof} [Proof of Corollary \ref{cor:nonempty} ]
	We want to show that condition in \eqref{eqn:nec_cond_p} implies nonemptyness of $C_{p}$. This is proven by construction. Note that $u(0)$ constructed in the proof of Theorem  \ref{thm:max_inv_set} only depends on the preview information $d_{1:p}(0)$. Let us denote $g:\R^{np} \rightarrow \R	$ a controller that maps a point $d_{1:p}(0)\in D^{p}$ to $u(0)$ we construct in the proof of Theorem \ref{thm:max_inv_set}. Then, let $x(0)$ be an arbitrary point in $ \mathbf{B}$ and let $d_{1:p}(t)\in D^{p}$ for $t \geq 0$. We pick control input $u(t) = g(d_{1:p}(t))$ for all $t \geq 0$. Then, we can verify that state $(x(n),d_{1:p}(n))$ at time $t=n$ satisfies the condition in Theorem \ref{thm:state_nec_cond}, that is $(x(n), d_{1:p}(n)) \in C_{p}$.
\end{proof}

\begin{proof}[Proof of Theorem \ref{thm:p_g_n}]
	Note that $k-j$ in \eqref{eqn:C_p_2} is less than $n$. Thus when $p \geq  n$, for all $k$, $2\leq k\leq n$ and all $j$, $1\leq j < k$, \eqref{eqn:C_p_2} becomes
	\begin{align}
	x_{k}+ \sum^{k-j}_{i=1} d_{i,k-i} \subseteq [b_{j,1}, b_{j,2}],
	\end{align}
	which is independent of $p$.  Thus, by definition of $C_{p}$, it is easy to check that 
	\begin{align}
	C_{p} = C_{n} \times D^{p-n}.
	\end{align}
\end{proof}

\begin{proof}[Proof of Theorem \ref{thm:two_bounds}]
	According to Theorem \ref{thm:p_g_n}, $$PROJ_{[1,n]}(C_{n}) = PROJ_{[1,n]}(C_{p}).$$ According to \eqref{eqn:upper_lower_bounds}, $PROJ_{[1,n]}(C_{n}) \subseteq C_{max,co}$.  The left of this proof is to show $C_{max,co} \subseteq PROJ_{[1,n] }(C_{n})$.
	
	Let $x\in C_{max,co}$. We want to show that there exists $d_{1:p}$ such that $(x, d_{1:p})\in C_{p}$. Since $C_{max,co}$ is a controlled invariant set of $\mathcal{D}(\Sigma_{B})$ within $ \mathbf{B}\times \R\times D$, for initial state $x(0):= x$ of $\mathcal{D}(\Sigma_{B})$, there exists  $(u(t))_{t=0}^{n-1}$ and $(u_{d})_{t=0}^{n-1}$ such that  $(x(t),u(t),u_{d}(t))\in \mathbf{B}\times \R\times D$ for all $t$ from $0$ to $n$. Based on this observation, we can easily verify that $(x,u_{d}(0), \cdots, u_{d}(n-1))$ satisfies constraints \eqref{eqn:state_nec_cond} and thus satisfies constraints \eqref{eqn:C_p_1} and \eqref{eqn:C_p_2} in Theorem \ref{thm:state_nec_cond} . Then by Theorem \ref{thm:max_inv_set}, $(x,u_{d}(0), \cdots, u_{d}(n-1) )\in C_{max,n}$.
\end{proof}

\balance
\end{document}